%% file: arXiv.tex
\documentclass[11pt]{article}
\usepackage{graphicx}

\title{Optimal Zonotope Containment via Adaptive Gaussian Sampling}

\date{}

\author{ 
Friedrich Eisenbrand \thanks{EPFL, Switzerland, \textsf{\{friedrich.eisenbrand, matteo.russo, ruben.skorupinski\}@epfl.ch}}
\and
Thomas Rothvoss  \thanks{ University of Washington, USA,  \texttt{\textsf{rothvoss@uw.edu}}}
\and
Matteo Russo\footnotemark[1] 
\and
Ruben Skorupinski\footnotemark[1]
}

\input{settings/packages}

\input{settings/macros}

\input{settings/environments}

\begin{document}

\maketitle

\begin{abstract}
    \noindent
    We investigate the convex-body containment problem
    $\max\{s >0 : s Z \subseteq Q\}$, where the outer body
    $Q \subseteq \R^d$ is described by a membership oracle and the
    inner body $Z \subseteq \R^d$ is a zonotope. Our main result is a
    sampling-based $O(\sqrt{d/\log d})$-approximation algorithm
    for this problem that matches the lower bound of
    $\Omega(\sqrt{d/\log d})$ by Khot and Naor in the oracle model.
    
    \smallskip
    \noindent Starting from approximate Lewis position, the algorithm samples
    Gaussian directions and adaptively shortens and repositions the generators.
    Its analysis combines preservation of the support in a fixed direction
    with contraction of the corresponding functional after renormalization.
    Crucially, the algorithm works directly with the input generators
    and does not use zonotope sparsification.
    We complement this result with a \emph{universal}
    $\Omega(\sqrt{d/\log d})$ lower bound holding for
    \emph{all} full-dimensional zonotopes when the number of oracle
    queries is polynomial in $d$.

    \smallskip
    \noindent
    Finally, we consider containment problems $\max\{s >0 : s K \subseteq Q\}$,
    for general convex bodies $K \subseteq \R^d$. A result of Nasz{\'o}di on
    approximating $K \subseteq \R^d$ by a polytope implies a
    $\Theta(d/\log d)$ approximation algorithm in polynomial time
    when the centroid of $K$ is the origin. We show the tightness of this
    approximation factor in the oracle model via a reduction to the
    circumradius computation. Our lower bound holds for centrally symmetric
    convex sets, implying that Barvinok's optimal $O(\sqrt{d})$-approximation of
    a centrally symmetric convex body by a polytope with a polynomial number
    of vertices cannot be computed in polynomial time.

\end{abstract}

\input{new_sections/01-introduction}
\input{new_sections/02-containment}

\input{new_sections/03-lowerbounds}

\input{new_sections/04-convexbodies}
\input{new_sections/05-polytopes}

\section*{Acknowledgments}

The work of Friedrich Eisenbrand and Ruben Skorupinski is supported by project 10000183 of the Swiss National Science Foundation (SNSF). Thomas Rothvoss is supported by NSF grant 2318620 \emph{AF: SMALL: The Geometry of Integer Programming and Lattices}. The authors used ChatGPT to assist with mathematical arguments. Specifically, the authors first designed an algorithm for zonotope containment based on Lewis position and normal-cone sampling that achieved a factor of $\sqrt{d}/\log^{1/4} d$. The authors then used ChatGPT Astra to build upon this idea and develop the adaptive rescaling algorithm in this paper that improves the earlier bound to the optimal $\sqrt{d/\log d}$. The authors take full responsibility for all the content of this work.

{\small
\printbibliography
}

\bigskip

{\noindent \LARGE \textbf{Appendix}}

\appendix

\input{new_sections/appendix}

\end{document}

%% file: settings/packages.tex
\usepackage[margin=0.95in]{geometry}
\usepackage{inconsolata}
\usepackage{libertine}
\usepackage[T1]{fontenc}
\usepackage[utf8]{inputenc}
\usepackage{enumerate}
\usepackage{utf8math}
\usepackage[tbtags]{amsmath}
\allowdisplaybreaks
\usepackage{amssymb,amsthm,mathtools, bbm}
\usepackage{physics}
\usepackage{hyperref}
\usepackage[svgnames]{xcolor}
\usepackage[capitalise,nameinlink]{cleveref}
\definecolor{links}{RGB}{11, 85, 255}
\definecolor{cites}{RGB}{0, 200, 0}
\definecolor{urls}{RGB}{255, 116, 0}
\hypersetup{
    colorlinks=true,
    linkcolor=black,
    citecolor=black,
    urlcolor=black
}
\usepackage[
    style=numeric, %
    natbib=true,      %
    doi=false,        %
    maxbibnames=99,   %
    maxcitenames=2,   
    url=false,        %
    isbn=false,       %
    backend=biber     %
]{biblatex}
\usepackage{doi}
\usepackage{tikz}
\usetikzlibrary{positioning, arrows.meta, shapes.geometric}
\usepackage{pgfplots}
\pgfplotsset{compat=1.14}
\usepgfplotslibrary{fillbetween}
\usepackage{hyphenat}
\usepackage[font=footnotesize]{caption}
\usepackage{subcaption}
\usepackage{appendix}
\usepackage{thm-restate}
\usepackage{nicefrac}
\usepackage{tcolorbox}
\usepackage{bm}

\usepackage{todonotes}

\usepackage{algorithm}
\usepackage[noend]{algpseudocode}

%% file: settings/macros.tex
\newcommand{\vol}{\textup{Vol}}
\newcommand{\poly}{\textup{poly}}
\newcommand{\conv}{\textup{conv}}

\newcommand{\Q}{\mathbb{Q}}
\newcommand{\R}{\mathbb{R}}

\newcommand{\E}[2]{\underset{#1}{\mathbb{E}}\left[#2\right]}
\newcommand{\B}{\mathbb{B}}
\renewcommand{\P}[2]{\underset{#1}{\mathbb{P}}\left[#2\right]}
\renewcommand{\exp}[1]{{\textup{exp}}\left(#1\right)}

\newcommand{\indep}{\perp\kern-5pt\perp}

\newcommand{\scalar}[1]{\left\langle#1\right\rangle}

%% file: settings/environments.tex
\theoremstyle{theorem}

\newtheorem{theorem}{Theorem}[section]
\newtheorem*{theorem*}{Theorem}
\newtheorem{lemma}[theorem]{Lemma}
\newtheorem{proposition}[theorem]{Proposition}
\newtheorem{corollary}[theorem]{Corollary}

\newtheorem{remark}[theorem]{Remark}

%% file: new_sections/01-introduction.tex
\section{Introduction}\label{sec:introduction}

 The \emph{containment problem} for convex bodies is as follows.
Given two convex bodies $K,Q \subseteq \R^d$ containing the origin in their relative interior, determine by how  much $K$ can be scaled while still being contained in $Q$.
Formally, it is the following optimization problem:
\begin{equation}
    \label{eq:containment-opt}
    (K,Q) \textsc{-Opt-Containment}:
    \max\{ \alpha >0 \mid \alpha  ⋅ K \subseteq Q\}.
  \end{equation}
If the outer body $Q$ is a polytope given in its inequality representation, then $(K,Q) \textsc{-Opt-Containment}$ can be solved in polynomial time simply by optimizing all normal vectors of $Q$ over $K$.
However, already when both bodies are polytopes, and the inner polytope is
described by inequalities
 while the outer one is described by its vertices, this problem is NP-hard \citep{freund1985complexity}. Here, we focus on approximating the optimal value $\alpha^\star$ of~\eqref{eq:containment-opt} within a factor $s(d)≥1$ that depends on the dimension $d$. Thus our task is to find $α∈ ℝ$ such that $α ≤\alpha^\star ≤ s(d) ⋅ α$ holds.

Via binary search, this approximation problem is  polynomial-time equivalent to  the following \emph{gap decision} version of the problem:
\begin{align}\label{eq:gap}
    (K,Q) \textsc{-Gap-Containment}: \begin{cases}
    \textsc{Yes,} &K \subseteq Q \\
    \textsc{No,}  &s(d) ⋅ K \nsubseteq Q
    \end{cases}.
\end{align}
An algorithm \emph{solves} the instance $(K,Q)$ of \textsc{Gap-Containment} if it correctly asserts whether it is a \textsc{Yes} or a \textsc{No} instance. It is not required to assert that both conditions hold, if this is the case.

Containment is closely related to classical geometric approximation tasks such as estimating radii, widths, and diameters of convex bodies, see, e.g., \citep{brieden1998approximation, Simonovits03}.
We study the containment problem
with a particular focus on the case in which the inner body  $K=Z$ is a zonotope and $Q$ is given by a membership oracle.

A \emph{zonotope} is a set of the form
\begin{equation}
    \label{eq:zonotope-def}
    Z(W) =
    \{Wx : x ∈ ℝ^n, \, \|x\|_\infty\le 1\},
\end{equation}
where $W \in \R^{d\times n}$ is the \emph{generator matrix} of the zonotope. The columns $w_1,\ldots,w_n \in \R^d$ of $W$ are the \emph{generators} of $Z(W)$.
Zonotopes are a well-studied class of convex polytopes  with rich structure and connections to functional analysis and convex geometry, see, e.g.~\citep{BourgainLM89}.
The \emph{zonotope containment} problem $Z ⊆Q$ has received considerable attention in the recent literature of fields such as control theory, neural networks or complexity theory \citep{althoff2016combining, kulmburg2021co, KulmburgSA25, FroeseGS25, FroeseGHS25, FroeseGHS26, Steinberg05, BhaskaraV11, BhattiproluGGLT23, GuthMU25}.

\subsection{Results of This Paper}
\label{sec:results}

We present approximation and hardness results for the
zonotope as well as for the  general containment problem. In passing, we conclude optimality and impossibility of two well known schemes to approximate a convex body by a polytope, respectively.
The \emph{computational model} that we use is the classical oracle model for
accessing convex bodies, i.e., closed convex sets with non-empty interior
\citep{GrotschelLS84}. We assume that the bodies contain the origin in their
interior. For the general-body sampling algorithm, we additionally assume
that the centroid of the inner body $K$ is the origin and that $K$ is
well-rounded, i.e., $r \B_2^d \subseteq K \subseteq R \B_2^d$, where
$0<r\le R$ and $R\in\poly(d)\cdot r$. The radii are given as part of the oracle
description. Zonotopes are given by rational generator matrices; a change of
coordinates puts them in the position needed by our algorithm. Our algorithms
operate with a \emph{membership} oracle which answers \textsc{Yes} or
\textsc{No} to the query ``$x\in K$'' for a point $x\in\Q^d$. Moreover, when the subscript is omitted, $\|\cdot\|$ denotes the $\ell_2$-norm.

\paragraph{Zonotope Containment.}
We provide a randomized polynomial-time algorithm that decides the
{\sc Gap-Containment} problem for $(Z,Q)$ with a factor of
$s(d)=O(\sqrt{d/\log d})$, where $Z\subseteq\R^d$ is a zonotope
(\Cref{thm:anc-containment}). This factor is optimal by the corresponding
lower bound of \citet{khot2007linear} for the special case in which
$Z$ is the hypercube. To the best of our knowledge, the $O(\sqrt{d/\log d})$ is also the best factor known to date even for the zonotope-in-zonotope problem, where both zonotopes are given explicitly. 

Our algorithm is based on the following principle: if the zonotope $Z$ is not contained in $Q$, then there exists an inequality $a^\top x \leq \beta$ that is valid for $Q$ but not for $Z$. This allows us to understand the containment problem as follows: an adversary picks an inequality $a^\top x \leq \beta$ not known to us and we are allowed to query whether a point $x^\star$ satisfies the inequality or not, as shown in \Cref{fig:placeholder}. Thus, the goal is to decide if such a point exists, whenever the inequality cuts sufficiently deep into the zonotope. \cref{alg:anc-run} allows us to do this whenever $\beta/h_Z(a) \leq c \sqrt{\log d/d}$ for a universal constant $c>0$. In the following we refer to the $a$ in the inequality as a functional.

The algorithm first puts the generators into approximate Lewis position.
Each run then performs $O(\log\log d)$ iterations: it samples a Gaussian
direction, shortens generators with large normalized projections onto that
direction, records a point in the input zonotope, and changes coordinates
using the matrix of the shortened generators. The analysis compares the
support retained in a fixed direction with the norm of the corresponding
functional after the change of coordinates. We show that, for every fixed
nonzero functional $a$, one run produces a point $x\in Z$ satisfying
\[
    a^\top x\ge c\sqrt{\frac{\log d}{d}}\,h_Z(a)
\]
with inverse-polynomial probability. Repeating independent runs and testing
the sampled points against the membership oracle of $Q$ gives the claimed
containment guarantee.

Crucially, our algorithm works directly with the input generators and
sidesteps zonotope generators sparsification. Combining the efficient
$O(d\log d)$-generator sparsification of \citet{CohenP15} with the sampling method inspired by the techniques of \citet{khot2007linear} gives an
$O(\sqrt d)$ approximation.
An efficient linear-size sparsifier would
give the optimal $O(\sqrt{d/\log d})$ factor by the same route.
Although \citet{ReisR26} recently established the existence of such
sparsifiers for constant approximation error, their construction is not
polynomial-time. Our algorithm obtains the optimal containment factor 
without requiring this construction.

\begin{figure}
    \centering
    \includegraphics[width=0.5\linewidth]{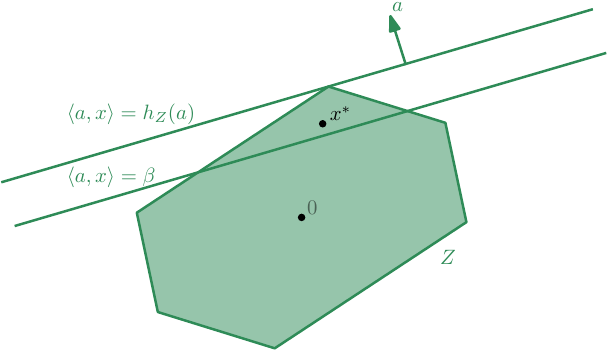}
    \caption{An adversary chooses an unknown inequality $a^\top x \le \beta$, and we test whether queried point $x^\star$ satisfies it.}
    \label{fig:placeholder}
\end{figure}

 \paragraph{Universal Hardness for Zonotopes.}
 \citet{khot2007linear} have shown a lower bound of  $s(d) = \Omega(\sqrt{d/\log d})$  for a randomized polynomial-time algorithm that decides {\sc Gap-Containment}, where the inner body is the hypercube.
 We show a \emph{universal} lower bound of $s(d) = \Omega(\sqrt{d/\log d})$ for the {\sc Gap-Containment} problem, where the inner body is \emph{any} full-dimensional zonotope and the number of oracle queries is polynomial in $d$.
 To achieve this result, we exploit the structural and volumetric properties of \emph{normalized} zonotopes introduced in \citep{bozzai2023vector}, and a connection between containment and polytopal approximation of convex bodies.

 \paragraph{Containment of General Convex Bodies.}
 \citet{Naszodi2019approximating}  provided a sampling based algorithm to  approximate any convex body $K$ whose centroid is the origin with a polytope $P$ spanned by a polynomial number of vertices, such that, if $P$ is scaled by $O(d/\log d)$, then it contains $K$ (see also~\citep{mustafa2022sampling}). One can further note that, for $s(d) = O(d/\log d)$, {\sc Gap-Containment}  is a {\sc Yes}-instance  if all vertices of $s(d)P$ are in $Q$, otherwise it is a {\sc No}-instance. Since sampling from a convex body can be done in polynomial time~\citep{DyerFK91,LovaszS93,CousinsV18}, then for $s(d) = O(d/\log d)$, {\sc Gap-Containment} can also be solved in randomized polynomial-time.
 A recent result by \citet{huang2026hardness} shows that this $O(d/\log d)$ factor is optimal for general (non-symmetric) convex bodies. This is achieved by constructing a specific convex body such that no polytope with a polynomial number of vertices can improve upon the $O(d/\log d)$ factor.  We show  a matching lower bound of  $s(d) = \Omega(d/\log d)$ for any polynomial-time algorithm for {\sc Gap-Containment} that holds even for \emph{symmetric} convex bodies. Consequently, Naszódi's method is an optimal algorithm for containment.
 
 \paragraph{Polytope Approximation.}
 For symmetric convex bodies $K ⊆ℝ^d$,
 \citet{Barvinok14} showed that there exists a polytope $P ⊆ K$ with a polynomial number of vertices such that $sP ⊇ K $  with  $s = O(\sqrt{d})$. Our lower bound proves that this approximation cannot be computed in polynomial time, if $K$ is given by a membership oracle. Moreover, in \cref{sec:containment-polytope}, we observe that the containment problem for a general body $Q$ is essentially equivalent to polytope approximation. Therefore, our guarantees carry over to this setting. This was already observed implicitly by \citet{brieden1998approximation} but we give a proof of this fact in \cref{app:omitted-polytopes} for completeness.

\subsection{Related Work}

The zonotope containment problem --- the question of whether one zonotope is fully contained in a second zonotope --- arises naturally in the context of complexity theory \citep{Steinberg05, BhaskaraV11, BhattiproluGGLT23, GuthMU25}, control theory \citep{althoff2016combining, kulmburg2021co, KulmburgSA25}, and neural network verification \citep{FroeseGS25, FroeseGHS25, FroeseGHS26}.
Despite its significance, the algorithmic literature on zonotope containment remains sparse. To the best of our knowledge, most existing algorithms focus on exact containment and, thus, require exponential time relative to the number of generators, or some other restrictive structural assumption \citep{SadraddiniT19, kulmburg2024search}. One exception is the polynomial time algorithm for zonotope containment provided in \citep{kulmburg2024approximability} which achieves an $O(\sqrt{n})$ approximation factor. We remark that, unless this algorithm is combined with zonotope sparsification \citep{CohenP15}, in which case it would become $O(\sqrt{d\log d})$, it would yield a bound depending only on the number of generators. On the contrary, all our bounds optimally depend on the ambient dimension $d$.
Furthermore, the scenario where a zonotope is contained within a convex body accessible only via an oracle remains largely unexplored, unlike the computation of inradius, circumradius, diameter, volume and other geometric functionals \citep{brieden1998approximation, DyerFK91, LovaszS90, LovaszS93, KannanLS97, LovaszV06, CousinsV18}.

\paragraph{Zonotope Sparsification.}
\citet{Talagrand90} showed that a zonotope in $\R^d$ can be approximated up to
constant factors by $O(d\log d)$ generators, and \citet{CohenP15} gave an
efficient construction using Lewis weights. The conjectured linear bound is
closely connected to the local theory of Banach spaces
\citep{bozzai2023vector,Rothvoss21}. For constant approximation error,
\citet{ReisR26} have recently obtained linear-size sparsifiers. Their construction is not
polynomial-time, and a polynomial-time construction for general zonotopes
remains open. Our containment algorithm achieves the optimal approximation
factor without invoking a sparsification routine.

\paragraph{Matrix Norms.}
Given two convex symmetric bodies $K$ and $Q$,  the $(K,Q) \textsc{-Opt-Containment}$ problem can be interpreted via a polar view: given two norms $\| \cdot \|_K$ and $\| \cdot \|_Q$, induced by convex bodies $K,Q \subseteq \R^d$, we seek to find the minimal scalar $s>0$ so that $\|x\|_Q \leq s \|x\|_K$ for all $x \in \R^d$. That is,
\[
    \forall x \in \R^d: ~\|x\|_Q \leq s \|x\|_K \iff \forall x \in \R^d: ~\|x\|_Q \leq \|x\|_{K/s} \iff \frac{1}{s}K \subseteq Q.
\]
Given the connection between approximating norms and the containment problem, we highlight that, when $K = A\B^n_q$ and $Q = \B^d_p$, $\max\{s>0: sK \subseteq Q\} = \nicefrac{1}{\|A\|_{q\to p}}$, where the $q\to p$ norm of a matrix $A \in \R^{d\times n}$ is defined by
\[
    \|A\|_{q \to p} = \max_{x \in \R^n \setminus \{0\}} \frac{\|Ax\|_p}{\|x\|_q}.
\]
Hence, approximating $(K,Q) \textsc{-Opt-Containment}$ for the specific bodies $K = A\B^n_q$ and $Q = \B^d_p$ is equivalent to approximating the $q \to p$ norm of the matrix $A$.
It is a well-known fact that, for all $p \ge 2$, $\ell_p$ balls $\B^d_p$ are \emph{zonoids} \citep[Theorem 6.6]{Bolker69}, where zonoids are centrally symmetric convex bodies that arise as limits of zonotopes \citep{BourgainLM89}. We note that the requirement of $p \ge 2$ is necessary as, for instance, $\B_1^d$ does not admit such a representation unless $d=2$. Thus,
the problem of zonoid containment entails approximating $q\to p$ norms as a special case as long as $p,q\ge 2$. We also stress that, by \citep[Theorem 4.12]{BhattiproluGGLT23} (originally \citep[Theorem 5]{schechtman1987more}), which states that any $\ell_p$-ball ($p\geq 2$) can be approximated in polynomial time by a zonotope with $O(n^3)$ many generators, approximating $q \to p$ norms can be phrased as a zonotope containment problem, but the converse need not hold.

The problem of computing $q \to p$ matrix norms in particular has been extensively studied: from a hardness perspective, \citet{BhaskaraV11} show that the problem of maximizing $p$ matrix norms is inapproximable within a factor of $O(2^{{(\log d)}^{1-\varepsilon}})$ for every $\varepsilon \in (0,1)$ unless $\textup{NP} \subseteq \textup{DTIME}(2^{\poly(\log d)})$, and \citet{BhattiproluGGLT23} identify conditions on $p,q$ for which this or other hardness of approximation results are possible. On the positive side, \citep{Steinberg05} also shows that there exists a polynomial time algorithm achieving an approximation factor of $\max\{d,n\}^{\nicefrac{25}{128}}$, later improved to $\max\{d,n\}^{3-2\sqrt{2}}$ by \citet{GuthMU25}. Using the efficient sparsification of \citet{CohenP15}, this yields an $O\left((d\log d)^{3-2\sqrt{2}}\right)$ approximation. Note that this is not a general zonotope containment problem since the outer body is required to be a $p$-ball.

\paragraph{Longest Vector-Sum Problem.} In the case where $K=Z(W)$ is a zonotope generated by matrix $W$ and $Q$ is any symmetric convex body, we recover the longest vector sum problem \citep{shenmaier2018approximability}. This problem aims to find a subset $S$ of vectors $V=\{v_1,\dots, v_n\} \subseteq \R^d$ such that $\|\sum_{v \in S}v\|_Q$ is maximized. The fact that the two problems are equivalent can be found in \citep[Lemma 1]{shenmaier2020complexity}. Due to the close relationship to the $(q\to p)$-norm problem, the lower bounds for the longest vector sum problem are essentially the same. Similarly, if $Q$ is an $\ell_p$-ball, approximation guarantees from the $(\infty \to p)$-norm problem carry over to this problem. The algorithms introduced in \citep{shenmaier2018approximability, shenmaier2020complexity} are almost exclusively exponential time algorithms. Indeed, the most recent result showcases a $(1-\varepsilon)$-approximation in time $O(d^{ O(1)}(1+2/\varepsilon)^d n)$. To the best of our knowledge, our adaptive sampling algorithm (\Cref{alg:anc-run}) is the first polynomial time algorithm for the general longest vector sum problem with an approximation factor better than $O(\sqrt{n})$, which was implicitly achieved by \citet{kulmburg2024approximability}.

%% file: new_sections/02-containment.tex
\section{Containment of a Zonotope in a Convex Body} \label{sec:cont_normal_cone}

In this section we prove our first main result on the containment problem where the inner body $Z \subseteq \R^d$ is a zonotope. We show that we can efficiently identify a point in $Z$ that for a fixed but unknown direction $a$ is not contained in the half-space $\{a^\top x \leq c\sqrt{\log d/d} \, h_Z(a)\}$. The algorithm is based on a membership oracle for this half-space. As already discussed in \cref{sec:results} this implies the general case.

\subsection{The Adaptive Sampling Algorithm}

We now describe a single run of the adaptive sampling \Cref{alg:anc-run}. In a pre-processing step, we first choose coordinates in which the generator directions are approximately in Lewis position. To be more precise, for given generators $V$, define
\begin{align*}
    \mathcal L(V)=\sum_{i=1}^n\frac{v_iv_i^\top}{\|v_i\|}.
\end{align*}
The identity $\mathcal L(V)=I$ describes exact Lewis position \cite[Section~2]{CohenP15}. We only need a constant-accuracy version: compute an invertible map $T$ such that $\frac{1}{4} I\preceq\mathcal L(TW)\preceq I$. This step is only performed once and can be done in polynomial time by approximating Lewis weights \cite[Section~3]{CohenP15} and rescaling by an absolute factor. We highlight that independent runs can reuse the same initial map $T$.
\par

For the adaptive sampling we define the map $q(t)$ which will be used to rescale generators based on their inner product with the sampled direction. We choose
\begin{align*}
    q(t)=\frac1{\max\{1,|t|\}},
\end{align*}
and discuss this choice and some alternatives in \Cref{sec:remarks}. Note that the function $tq(t)$, which is used to compute the sample $x_{\text{orig}}$, just clips all values $t>1$ to 1, the values $t<-1$ to $-1$ and keeps the others unchanged. 
In each iteration, the algorithm draws a Gaussian direction and first records a point in $Z$, obtained by applying the clipping map $tq(t)$ to the standardized projections of the current generators and converting back to the original coordinates using $A$. It then shortens the generators using the weights $q(t)$ and repositions them using the matrix $S$, preparing the coordinates for the next iteration. The map $A$, initially $T^{-1}$, is updated accordingly. The scale $\lambda$ controls the clipping threshold and is halved after each iteration, so each run lasts for $O(\log\log d)$ iterations. Throughout the algorithm we use a fixed absolute constant $C\ge64$ and assume that $d\ge3$.

\begin{algorithm}[H]
\caption{One run of adaptive sampling}\label{alg:anc-run}
    \begin{algorithmic}[1]
    \State Set $v_i\gets Tw_i$ for every $i$, $A\gets T^{-1}$, and $\lambda \gets\log d$. \Comment{Approximate Lewis position}
    \While{$\lambda\ge1$}
        \State Draw $z\sim N(0,I_d)$ independently of all previous draws.
        \State $t_i\gets\langle z,v_i\rangle/(C\lambda\|v_i\|)$, \quad $q_i\gets q(t_i)$
               for every $i$.
        \State Record $x_{\rm orig}=A\sum_it_iq_iv_i$.
        \State $S\gets\sum_i q_i v_iv_i^\top/\|v_i\|$.
        \State $v_i\gets q_iS^{-1/2}v_i$ for every $i$. \Comment{Rescale and reposition generators}
        \State $A\gets AS^{1/2}$;\quad $\lambda\gets \lambda/2$.
    \EndWhile
    \State \textbf{Return} all recorded points.
    \end{algorithmic}
\end{algorithm}

\begin{proposition}[Single-run guarantee]\label{prop:clipped-run}
    Let $a \in \R^d$ be a direction. There exists an absolute constant $c >0$ such that at least one of the $O(\log\log d) $ samples $x_{\rm orig} \in Z$, satisfies 
    \begin{align*}
         \langle a,x_{\rm orig}\rangle \ge c\sqrt{\frac{\log d}{d}}\,h_Z(a).
    \end{align*}
    with probability at least $c/(d\log^4 d)$.
\end{proposition}

Once this proposition is established, the main theorem follows directly. Indeed, assume there exists a point $x^\star \in Z \setminus Q$. Then, by convexity, there exists a hyperplane $\{x:\langle x,a\rangle = b\}$ separating $x^\star$ from $Q$. We may choose $x^\star$ to be a vertex maximizing $\langle a,x\rangle$ over $Z$, so $\langle a,x^\star\rangle=h_Z(a)$. \Cref{prop:clipped-run} implies that if $Z$ is scaled by a factor of $1/c \, \sqrt{d/\log d}$ then \cref{alg:anc-run} returns a point $x_{\rm orig}$ satisfying $\langle a,x_{\rm orig}\rangle\ge h_Z(a)>h_Q(a)$, and hence lying outside $Q$ with probability $c/(d\log^4 d)$. Repeating this independently polynomially many times and rescaling the inner body yields the guarantees of the theorem below. 

\begin{theorem}\label{thm:anc-containment}
    There is an absolute constant $\eta > 0$ such that, for every $0<\delta<1$, testing the points from $O(d\log^4 d\log(1/\delta))$ independent runs of \Cref{alg:anc-run} with a membership oracle for $Q$ either finds a point of $Z\setminus Q$ or reports
    \begin{align*}
        Z\subseteq \frac{2}{\eta}\sqrt{\frac d{\log d}}\,Q,
    \end{align*}
    where the probability of incorrectly reporting an inclusion is at most $\delta$.
\end{theorem}

To understand the main idea of the proof of \Cref{prop:clipped-run} we first consider one iteration of \Cref{alg:anc-run}. We denote by $v_i^+ = q_iS^{-1/2}v_i$ the updated generators, by $K=\sum_i[-v_i,v_i]$, $K^+=\sum_i[-v_i^+,v_i^+]$ the current and updated zonotopes and by $b=A^\top a$ and $b^+ = S^{1/2}b$ the current and updated functionals. Furthermore, denote by $H= h_K(b)$ and $H^+=h_{K^+}(b^+)$ the support functions of $K$ and $K^+$. With this notation at hand, observe that we can express the squared norm $\|b^+\|^2$ and the support-loss in terms of the weights $q_i$ and the matrix $S$. Indeed we have that 
\begin{align*}
    \|b^+\|^2 = b^\top Sb, \qquad H-H^+ = \sum_{i=1}^n (1-q_i)|\langle b,v_i \rangle|.
\end{align*}

The main ingredient of the proof is \cref{lem:good_update}, which relates the probability of sampling a point of height at least $c\sqrt{\log d/d} \, h_Z(a)$ to the evolution of these two quantities. Informally, unless obtaining such a point is already sufficiently likely, there is a controlled probability of an update that loses little support and satisfies a strong bound on the new norm. We call such updates \emph{good}. The probability bounds are proven by conditioning on an event $\mathcal E$ that controls how aligned the random vector $z$ and $a$ are. The choice of scales ensures that, if this alternative persisted throughout the run, a sequence of good updates would keep the support bounded from below by a constant while forcing the norm to become too small. Geometrically, the support in the unit direction $b/\|b\|$, namely $H/\|b\|$, would become larger than $\sqrt d$. This contradicts the invariant $K\subseteq\sqrt d\,\B_2^d$, which follows from the initial upper Lewis bound and our choices of $S$ and the weights $q_i$. Thus, along such a good sequence, we must reach an iteration where the desired sample is sufficiently likely. The quantitative probability bounds in the local lemma ensure that following the good updates and then obtaining this sample has inverse-polynomial probability in $d$.

\subsection{Proof of the single-run guarantee}\label{sec:clipped-run}

For the formal proof of \Cref{prop:clipped-run} we will rely on several technical lemmas. To ease notation in the analysis of one iteration, we make use of local coordinates. We denote by $v_i$ and $v_i^+$ the generators before and after the entire iteration, respectively. Set $K=\sum_i[-v_i,v_i]$, $K^+=\sum_i[-v_i^+,v_i^+]$, let $q_i=q(\langle z,v_i\rangle/(C\lambda\|v_i\|))$ be the rescaling factors of the individual generators. Then
\begin{align*}
     S=\sum_{i=1}^n q_i\frac{v_iv_i^\top}{\|v_i\|},\qquad
     x=\frac{Sz}{C\lambda}
       =\sum_{i=1}^n q_i\frac{\langle z,v_i\rangle}{C\lambda\|v_i\|}v_i,\qquad
     v_i^+=q_iS^{-1/2}v_i.
\end{align*}

Since $q_i>0$, the generators continue to span $\mathbb R^d$ and $S$ is positive definite. Moreover, $x\in K$ because $|tq(t)|\le1$. The map $A$ takes current coordinates to original coordinates. Initially $Av_i=w_i$ is given by the approximate Lewis transform, and each update satisfies $A^+v_i^+=q_iAv_i$. Hence $Av_i$ is always a positive multiple, at most one, of $w_i$, so $AK\subseteq Z$ and $x_{\rm orig}=Ax\in Z$. For a fixed original functional $a$, write $b=A^\top a$. Its update is $b^+=S^{1/2}b$, and $\langle a,x_{\rm orig}\rangle=\langle b,x\rangle$.
\par

The first short lemma stated below is a simple bound on the width of the zonotope $K$ in every iteration of the algorithm that follows from the initial Lewis position.

\begin{lemma}
    The inequality $\mathcal L(V)\preceq I$ holds at every iteration. Thus we have $h_K(b) \leq \sqrt{d}\|b\|$ and initially $h_{K_0}(b_0)\ge\|b_0\|/4$. Here $K_0$ denotes the zonotope in approximate Lewis position at the beginning of the algorithm and $b_0$ the corresponding functional.
\end{lemma}

\begin{proof}
    Since the matrix $S$ satisfies $S\preceq\mathcal L(V)\preceq I$ this directly implies $\|S^{-1/2}v_i\|\ge\|v_i\|$, and consequently    
    \begin{align*}
         \mathcal L(V^+) =S^{-1/2}\!\left(\sum_iq_i \frac{v_iv_i^\top}{\|S^{-1/2}v_i\|}\right)S^{-1/2} \preceq S^{-1/2}SS^{-1/2}=I.
    \end{align*}
    Taking the trace gives $\sum_i\|v_i\|\le d$. Thus, for every $b$, Cauchy--Schwarz yields $h_K(b)^2 \le d\|b\|^2.$ Initially, for every unit $u$, $h_K(u)\ge\sum_i\langle u,v_i\rangle^2/\|v_i\|\ge1/4$ holds and hence, in particular, $h_{K_0}(b_0)\ge\|b_0\|/4$.
\end{proof}

We will now define the event $\mathcal{E}$ which essentially controls how aligned the sample $z$ and the normalized current functional $b$ are, and plays an important role in the following analysis. More precisely, we define
\begin{align*}
     \mathcal E=\left\{\frac{\sqrt{\log d}}{C\lambda}\le \left\langle z,\frac b{\|b\|}\right\rangle \le\frac{\sqrt{\log d}}{C\lambda}+1\right\}, \qquad \mu=\mathbb E[\langle b,x\rangle\mid\mathcal E].
\end{align*}
Observe that, if we had $\mu \geq  \Omega(\sqrt{\log d/d})h_Z(a)$ we could immediately deduce the desired probability bound for \cref{prop:clipped-run} using the reverse Markov inequality combined with the law of total probability and $\P{}{\mathcal{E}} \geq 1/\poly (d)$. This is formalized below in the proof of the proposition. Hence, the remainder of the proof consists of showing that such a ``large-mean-event'' will happen with at least inverse-polynomial probability. For this we rely on the following lemma, defining good updates. As before, we use the local variables $H=h_K(b)$ and $H^+=h_{K^+}(b^+)$.

\begin{lemma}[Good updates]\label{lem:good_update}
    There is an absolute constant $c_0>0$ such that $\P{}{\mathcal{E}}\ge c_0\exp{-\log d/(C^2\lambda^2)}$.
    Moreover, $\mu>0$, and conditional on $\mathcal E$, with probability at least $1/2$, both inequalities
    \begin{equation}\label{eq:clipped-good}
         \|b^+\|^2\le\frac{8C^2\lambda^2}{\sqrt{\log d}}\,\|b\|\mu,
         \qquad
         H-H^+\le\frac{4H}{C\lambda}+4\mu
    \end{equation}
    hold. We call an update satisfying~\eqref{eq:clipped-good} \emph{good}.
\end{lemma}

Before proving the lemma, we show how \Cref{prop:clipped-run} follows from it. The main idea is that a sequence of good events will always need to include a ``large-mean-event'' since otherwise the bounds of \Cref{lem:good_update} eventually become contradictory. Combining the probabilities of good updates with the probability of obtaining the desired sample once the mean is large yields the inverse-polynomial lower bound on the overall success probability.

\begin{proof}[Proof of \Cref{prop:clipped-run}]
    Let $\lambda_j=(\log d)/2^j$, and let $R=O(\log \log d)$ be the number of iterations, indexed by $j=0,\ldots,R-1$. Thus $\lambda_{R-1}\ge1>\lambda_R$ and $2^R>\log d$. Fix $a\ne0$ and assume without loss of generality that $h_Z(a)=1$. Let $K_j,A_j,b_j$ denote the state before iteration $j$, with the subscript $R$ denoting the final state, and put $H_j=h_{K_j}(b_j)$. Initially, $H_0=1$ and because $K_0$ is in approximate Lewis-position $\|b_0\|\le4$. For each realized history, let $\mathcal E_j$ and $\mu_j$ be the event and conditional mean defined above for the resulting current state. In particular, $\mu_j$ is determined before the next Gaussian is drawn. Fix a sufficiently small absolute $\eta>0$, after fixing $C$. 
    \par
    Call a state \emph{small-mean} if $\mu_j<\eta\sqrt{\log d/d}$ and suppose that all $R$ states are small-mean and all $R$ updates are good. Note that support decreases only through shortening, so $H_j\le H_0=1$ and since $\sum_{j<R}\lambda_j^{-1}<2$, summing the support-loss inequalities and applying \Cref{lem:good_update} gives
    \begin{align*}
     H_R = H_0-\sum_{j=0}^{R-1}(H_j-H_{j+1}) \ge 1-4\sum_{j=0}^{R-1} \left(\frac{H_j}{C\lambda_j}+\mu_j\right) \ge 1-\frac8C-4\eta R\sqrt{\frac{\log d}{d}} \ge \frac34.
    \end{align*}
    The last inequality holds for sufficiently small $\eta$, for all $d\ge3$, because $R\sqrt{\log d/d}$ is bounded.
    
    On the other hand, the norm inequalities of \Cref{lem:good_update} combined with the small-mean assumption give $\|b_{j+1}\|^2\le8C^2\eta \lambda_j^2\|b_j\|/\sqrt d$. Dividing by $\lambda_{j+1}^4=\lambda_j^4/16$ and iterating yields the chain
    \begin{align*}
         H_R\le\sqrt d\,\|b_R\| \le \lambda_R^2(128C^2\eta)^{1-2^{-R}} \left(\frac{\sqrt d\,\|b_0\|}{\lambda_0^2}\right)^{2^{-R}} \le 64C\sqrt\eta<\frac12,
    \end{align*}
    where we used that $\|b_0\|\le4$, $\lambda_R=(\log d)/2^R<1$, and chose $\eta$ sufficiently small. Therefore, we deduce by contradiction that no complete run can consist entirely of good updates from small-mean states. 
    
    \par
    
    It remains to prove that a run of good updates is sufficiently likely. For this, call an output \emph{successful} if its height is at least $\eta/2 \, \sqrt{\log d/d}$. All the following probabilities are conditional on the realized history. At a large-mean state, $\langle b,x\rangle\le h_Z(a)=1$ implies
    $\langle b,x\rangle\le \mu/2+\mathbf 1_{\{\langle b,x\rangle\ge\mu/2\}}$, and hence, taking expectations conditioned on $\mathcal E$ results in 
    \begin{align*}
         \P{}{\langle b,x\rangle\ge\frac\mu2\mid\mathcal E}\ge\frac\mu2\ge\frac\eta2\sqrt{\frac{\log d}{d}}.
    \end{align*}
    Write $p_j=c_0/2 \, \text{exp}(-\log d/(C^2\lambda_j^2))$. Conditional on the history so far, a small-mean state at iteration $j$ admits a good update with probability at least $p_j$ by \Cref{lem:good_update}, while a large-mean state yields a successful sample with probability at least $p_j\eta/2 \, \sqrt{\log d/d}$. The contradiction above ensures that following good updates from small-mean states must reach a large-mean state before the run ends. Multiplying these conditional probability bounds up to that iteration, and using $p_j\le1$ to extend the product to all $R$ iterations, gives success probability at least
    \begin{align*}
         \frac\eta2\sqrt{\frac{\log d}{d}} \prod_{j=0}^{R-1}p_j \geq \left(\frac{c_0}{2}\right)^R \frac\eta2\sqrt{\frac{\log d}{d}} \exp{-\frac{4\log d}{3C^2}}\ge\frac{c}{d\log^D d}
    \end{align*}
    for absolute $c,D>0$. Here, we additionally used that $\sum_{j<R}\lambda_j^{-2}\le4/3$, $1/2+4/(3C^2)<1$ and $(c_0/2)^R\ge c'(\log d)^{-D}$. Picking $c_0=1/8$ allows the choice $D=4$, finishing the proof.
\end{proof}

It remains to establish \cref{lem:good_update}. For this, the most important observation is that the normalized inner product $\langle z, v_i \rangle/\|v_i\|$ conditioned on $\mathcal{E}$ essentially behaves like a shifted normal distribution. Since $\|b^+\|^2 = b^\top Sb = \sum_i q_i \langle v_i,b\rangle^2/\|v_i\|$ and $H-H^+ = \sum_i (1-q_i)|\langle v_i,b\rangle|$ the proof of the lemma thus reduces to understanding $\E{}{q(Y/L)}$ where $Y$ is a shifted normal and $L$ a constant ($L$ will be at the scale of $C\lambda$). This is the content of the following technical lemma. Its proof was deliberately kept as general as possible to allow further discussions about the choice of the scaling function $q(t)$. Indeed, the proof works for many different versions of $q(t)$ which will be further elaborated in \Cref{sec:remarks}. For now, recall that we defined $q(t)=1/{\max\{1,|t|\}}$.

\begin{lemma}\label{lem:clipped-scalar}
Let $Y=m+\sigma N$, where $N\sim N(0,1)$, $m\ge0$,
$0\le\sigma\le1$, and $L\ge1$. Then
\[
 m\,\mathbb E q(Y/L)\le2\,\mathbb E[Yq(Y/L)],
 \qquad
 \mathbb E[1-q(Y/L)]
 \le\frac{1+\mathbb E[Yq(Y/L)]}{L}.
\]
\end{lemma}

\begin{proof}
The case $\sigma=0$ follows directly from $1-q(t)\le |t|q(t)$, so assume $\sigma>0$. A Gaussian with nonnegative mean has density at $y>0$ at least its density at $-y$. Pairing opposite points therefore shows that every integrable odd function nonnegative on the positive half-line has nonnegative expectation.
The function $q$ is even, Lipschitz thus absolutely continuous and $-q'(t)\le tq(t)$ almost everywhere on $t>0$. Applying the preceding observation to $tq(t)$ and $tq(t)+q'(t)$ (both are odd) gives $\mathbb E[Yq(Y/L)]\ge0$ and $\mathbb E[-q'(Y/L)]\le( 1/L)\mathbb E[Yq(Y/L)]$. Gaussian integration by parts, valid for absolutely continuous functions, now yields the first inequality.
\begin{align*}
    m \mathbb E q(Y/L) =\mathbb E[Yq(Y/L)] +\frac{\sigma^2}{L}\mathbb E[-q'(Y/L)] \le\left(1+\frac{\sigma^2}{L^2}\right)\mathbb E[Yq(Y/L)] \le2\,\mathbb E[Yq(Y/L)]
\end{align*}
For the second inequality, set $(x)_+=\max\{x,0\}$ and note that then $\mathbb E(-Y)_+\le\sigma\mathbb E(-N)_+ \le1/\sqrt{2\pi}<1/2$. Therefore, using $1-q(t)\le |t|q(t)$ and $q(t)\le1$ we deduce 
\begin{align*}
    \mathbb E[1-q(Y/L)] \le\frac1L\mathbb E[|Y|q(Y/L)] =\frac1L\mathbb E[Yq(Y/L)] +\frac2L\mathbb E[(-Y)_+q(Y/L)]  \le\frac{\mathbb E[Yq(Y/L)]+1}{L},
\end{align*}
which finishes the proof.
\end{proof}

With this lemma at hand, we conclude with the proof of \Cref{lem:good_update}.

\begin{proof}[Proof of \Cref{lem:good_update}]
Let $u=\langle z,b/\|b\|\rangle$. For $t=\sqrt{\log d}/(C\lambda)$, the standard Gaussian density on $[t,t+1]$
is at least $e^{-1}e^{-t^2}/\sqrt{2\pi}$, since $(t+1)^2/2\le t^2+1$. This proves the asserted lower bound on $\Pr(\mathcal E)$.

Since $q$ is even, reversing a generator changes neither $S$ nor $x$, so we may orient the generators
such that $\langle b,v_i\rangle\ge0$. Put $y_i=\langle b,v_i\rangle/(\|b\|\|v_i\|)\in[0,1]$ and condition on a value $u\in[t,t+1]$. Then, the normalized projection of $z$ onto the $i$-th generator $v_i$ has the following conditional distribution:
\begin{align*}
     \frac{\langle z,v_i\rangle}{\|v_i\|} \stackrel{\rm law}=\ u y_i+\sqrt{1-y_i^2}\,N,
\end{align*}
where $N\sim \mathcal{N}(0,1)$ denotes a standard normal.

Now, apply \Cref{lem:clipped-scalar} with $m=uy_i$ and $L=C\lambda$. Multiply its first inequality by $\|b\|\langle b,v_i\rangle/u$, its second by $\langle b,v_i\rangle$, and sum over $i$. The update identities
$\|b^+\|^2=b^\top Sb$ and $H-H^+=\sum_i(1-q_i)\langle b,v_i\rangle$ then give
\begin{align*}
    \mathbb E[\|b^+\|^2\mid u] \le\frac{2C\lambda}{u}\|b\| \mathbb E[\langle b,x\rangle\mid u], \qquad 
    \mathbb E[H-H^+\mid u] \le\frac{H}{C\lambda}+\mathbb E[\langle b,x\rangle\mid u].
\end{align*}
The scalar estimate shows that each summand in the conditional mean is nonnegative. At least one is positive, because $b\ne0$, the generators span $\mathbb R^d$, $u>0$, and $q>0$. Thus $\mu>0$. Averaging over $u$ conditional on $\mathcal E$, and using
$u\ge\sqrt{\log d}/(C\lambda)$, bounds the two expectations by $2C^2\lambda^2\|b\|\mu/\sqrt{\log d}$ and $H/(C\lambda)+\mu$, respectively. Both quantities being bounded are nonnegative. Markov's inequality with factor 4 bounds the failure probability of each inequality in~\eqref{eq:clipped-good} by $1/4$. Finally, a union bound yields the conditional probability of 1/2 which completes the proof.
\end{proof}

The same method gives a sampling guarantee at every height. Its proof, deferred to \Cref{app:zonotope-tails}, uses the same generator and coordinate updates with different weights and an adjusted initial scale, and forms partial sums of the recorded points after the run.

\begin{restatable}[Sampling at every height]{corollary}{zonotopetails}\label{cor:zonotope-tails}
There exist absolute constants $C_1,C_2,C_3>0$ such that, for every full-dimensional zonotope $Z=\sum_i[-w_i,w_i]\subseteq\R^d$, $d\ge3$, a randomized polynomial-time algorithm samples from a distribution $\mathcal D$ over $Z$ satisfying
\[
 \P{X\sim\mathcal D}{\langle a,X\rangle\ge
       \frac{C_1s}{\sqrt d}\,h_Z(a)}
 \ge C_2\exp{-C_3s^2}.
\]
for all $a\in\R^d,\ 0\le s\le\sqrt d$.
\end{restatable}

\subsection{Some Further Remarks} \label{sec:remarks}

\begin{remark}[The choice of $q(t)$]
    As we have mentioned during the proof of \Cref{prop:clipped-run}, there is significant freedom in how the scaling function $q$ is chosen. Indeed, the proof applies unchanged to any even, locally absolutely continuous function $q$ satisfying, for $t>0$,
    \begin{align*}
        \frac1{1+t}\le q(t)\le\min\left\{1,\frac 1t\right\},
        \qquad -q'(t)\le tq(t),
    \end{align*}
    where the derivative inequality is required almost everywhere. The function being even implies that the generator orientation is irrelevant and the pointwise bounds ensure positive shortening, feasible outputs, and give $1-q(t)\le tq(t)$ for the support-loss estimate. The derivative condition is used in the Gaussian scalar estimate. Besides our choice of $q(t)=1/{\max\{1,|t|\}}$, for example the class of functions 
    \begin{align*}
        q(t)=(1+|t|^\alpha)^{-1/\alpha}\quad(\alpha\ge2)
    \end{align*}
    also satisfies these properties and hence could be used in the algorithm.
\end{remark}

\begin{remark}[Original vertices as outputs]
    We remark that instead of recording $x_{\text{orig}}$ by keeping track of the transformation matrix $A$ in \cref{alg:anc-run}, we can obtain the same guarantees by simply recording the vertex of the original zonotope given by
\begin{align*}
    x_{\rm vert}=\sum_{i=1}^n w_i\operatorname{sign}\langle z,v_i\rangle,
\end{align*}
where $v_i$ are the generators at the start of the current iteration.
Note that this is indeed a vertex of $Z$ since throughout the algorithm $Av_i=\lambda_iw_i$ for some $0<\lambda_i\le1$, so the displayed signs are those of $\langle A^{-\top}z,w_i\rangle$. Almost surely none of these inner products vanishes, and $x_{\rm vert}$ is the unique maximizer of this functional over $Z$.
To verify the guarantee, orient the generators $v_i$ and $w_i$ so that $\langle a,w_i\rangle\ge0$. Then, with $b=A^\top a$, we deduce that $0\le\langle b,v_i\rangle\le\langle a,w_i\rangle$. The Gaussian-pairing argument above gives $0\le\mathbb E[t_iq_i\mid\mathcal E] \le\mathbb E[\operatorname{sign} (t_i)\mid\mathcal E]$. Together with the coefficient bounds, this yields
\begin{align*}
    \mu_{\rm vert} :=\mathbb E[\langle a,x_{\rm vert}\rangle\mid\mathcal E] \ge\mathbb E[\langle b,x\rangle\mid\mathcal E] =\mu,
\end{align*}
where $x=\sum_it_iq_iv_i$ is the sample in the current coordinates. Replacing $\mu$ throughout the analysis by $\mu_{\rm vert}$ only weakens the two good-update inequalities. The support and norm
argument is therefore unchanged, and the large-mean success bound follows from $\langle a,x_{\rm vert}\rangle\le h_Z(a)=1$ after normalization.
\end{remark}

%% file: new_sections/03-lowerbounds.tex
\section{Universal Lower Bounds for Zonotope Containment} \label{sec:universal}

In this section, we show that, in the oracle model, \emph{for all} zonotopes, one cannot improve on the $\Omega(\sqrt{d/\log d})$ factor for containment. We emphasize that our result is a \emph{universal} lower bound and, as such, it is not implied by the lower bound in \cite{khot2007linear}, where the authors show that there \emph{exists} a zonotope---namely, the hypercube---for which this factor is tight. We do, however, apply a similar strategy to deduce our lower bound, relating containment to approximation of $Z$ by a polytope $P$. The formal relation between containment and polytope approximation is discussed in \cref{sec:containment-polytope}.

To prove the universal lower bound for polytope approximation, we place the zonotope in (exact) Lewis position. We may do that without loss of generality since both the approximation factor and the volume ratio are invariant under invertible linear transformations. We first record the two bounds needed below.

\begin{lemma}\label{lem:lewis-support-volume}
Let $Z=\sum_i[-v_i,v_i]\subseteq\R^d$ have nonzero generators satisfying $\mathcal L(V) = \sum_i v_iv_i^\top/\|v_i\|=I$. Then
\[
 \B_2^d\subseteq Z\subseteq\sqrt d\,\B_2^d,
 \qquad \vol(Z)\ge2^d.
\]
\end{lemma}

\begin{proof}
Taking traces in the Lewis identity gives $\sum_i\|v_i\|=d$. Thus, for every unit vector $a$, using $|\langle a,v_i\rangle|\le\|v_i\|$ and Cauchy--Schwarz,
\[
 1=\sum_{i=1}^n\frac{\langle a,v_i\rangle^2}{\|v_i\|}
 \le h_Z(a)
 \le\left(\sum_{i=1}^n\|v_i\|\right)^{1/2}
      \left(\sum_{i=1}^n\frac{\langle a,v_i\rangle^2}{\|v_i\|}\right)^{1/2}
 =\sqrt d.
\]
For the volume bound, apply \cite[Lemma~4]{ball1991shadows} with $u_i=v_i/\|v_i\|$ and $a_i=c_i=\|v_i\|$. The Lewis identity gives $\sum_i c_iu_iu_i^\top=I$, and hence
\[
 \vol(Z)\ge2^d\prod_{i=1}^n \left(\frac{a_i}{c_i}\right)^{c_i}=2^d. \qedhere
\]
\end{proof}

With these estimates at hand we can now prove the lower bound, by arguing that any polytope $P \subseteq Z$ with polynomially many vertices cannot even approximate the volume of any zonotope $Z$ well. It then follows directly that the containment cannot be well-approximated either. We remark that for the specific case of $P= \B_1^d$ this was already established in \cite{schneider2001busemann} where a lower bound of $\Omega(\sqrt{d})$ on the Banach-Mazur-distance between any zonotope and $B_1^d$ was proven.

\begin{theorem}\label{thm:appr_zono_by_poly_error}
Let $Z\subseteq\R^d$ be a zonotope and let $P=\conv\{x_1,\ldots,x_N\}\subseteq Z\subseteq sP$ be a polytope on $N=\poly(d)\geq 2d$ vertices approximating $Z$ up to a factor $s=s(d)$. Then
\[
 s\in\Omega\left(\sqrt{\frac d{\log (N/d)}}\right).
\]
\end{theorem}

\begin{proof}
By affine invariance, assume that $Z$ is in Lewis position. By \Cref{lem:lewis-support-volume}, $P/\sqrt d\subseteq\B_2^d$. The standard volume bound for inscribed polytopes gives
\[
 \left(\frac{\vol(P/\sqrt d)}{\vol(\B_2^d)}\right)^{1/d}
 \le O\left(\sqrt{\frac{\log (N/d)}d}\right).
\]
For example, apply \cite[Proposition~1.7]{brieden2001deterministic} to the symmetric hull of the scaled vertices, which has at most $2N$ vertices and is still contained in $\B_2^d$. Using \Cref{lem:lewis-support-volume} and $\vol(\B_2^d)^{1/d}\le\sqrt{2\pi e/d}$, we obtain
\[
 \left(\frac{\vol(P)}{\vol(Z)}\right)^{1/d}
 \le\frac{\sqrt d\,\vol(\B_2^d)^{1/d}}2
       \left(\frac{\vol(P/\sqrt d)}{\vol(\B_2^d)}\right)^{1/d}
 \le O\left(\sqrt{\frac{\log (N/d)}d}\right).
\]
Since $\vol(Z)\le s^d\vol(P)$, this proves the theorem.
\end{proof}

The volume-ratio bound holds for every polynomial-vertex polytope $P\subseteq Z$, without assuming $Z\subseteq sP$. Together with \Cref{prop:containment-polytope} in \Cref{sec:containment-polytope}, this proves the claimed universal oracle lower bound for both deterministic and randomized containment algorithms using polynomially many queries in $d$.

\begin{remark}
The lower bound is tight for all zonotopes by \Cref{thm:anc-containment} and the reductions in \Cref{prop:containment-polytope}, which give, with high probability, a polytope $P$ with polynomially many vertices satisfying $P\subseteq Z\subseteq O(\sqrt{d/\log d})P$. Consequently, containment approximation and volumetric approximation by polynomial-vertex inscribed polytopes have the same optimal order $\Theta(\sqrt{d/\log d})$ for zonotopes.
\end{remark}

\begin{remark}
For every centrally symmetric isotropic convex body $K\subseteq\R^d$ and every $\Omega(d)\le N\le O(\exp d)$ in the precise range $c_0d\le N\le e^d$, where $c_0>1$ is a sufficiently large absolute constant, let $P=\conv(\pm x_1,\ldots,\pm x_N)$ for independent uniform samples $x_i$ from $K$. Then
\[
 \mathbb E\left[\left(\frac{\vol(P)}{\vol(K)}\right)^{1/d}\right]
 \in\Theta\left(L_K\sqrt{\frac{\log(N/d)}d}\right)
 =\Theta\left(\sqrt{\frac{\log(N/d)}d}\right),
\]
where $L_K:=\det(\mathrm{Cov}(K))^{1/(2d)}/\vol(K)^{1/d}$ is the isotropic constant.
The bounds of \citet[Equations (1.7)--(1.8)]{dafnis2013quermassintegrals}, together with $L_K\in\Theta(1)$, give this order with probability at least $1-1/N$. Since $P\subseteq K$, the random volume ratio to the power $1/d$ lies in $[0,1]$, which also gives the stated expectation estimate. The uniform bound on $L_K$ follows from the affirmative resolution of Bourgain's slicing conjecture by Klartag and Lehec using Guan's bound \citep{guan2024note,klartag2025affirmative}.
Taking, for example, $N=d^2$, the preceding volumetric lower bound shows that this approximation error is best possible, up to absolute factors, for zonotopes. Thus uniform sampling produces an order-optimal volumetric approximation with polynomially many vertices. This is false for general convex bodies: for the cross-polytope $\B_1^d$, the same sampling budget gives only the $\Theta(\sqrt{\log d/d})$ volume-radius ratio, whereas $\B_1^d$ itself has $2d$ vertices.

\end{remark}

%% file: new_sections/04-convexbodies.tex
\section{Containment of General Convex Bodies}\label{sec:general}

We now consider $(K,Q)\textsc{-Gap-Containment}$ for general convex bodies $K,Q\subseteq\R^d$. For the sampling upper bound, we assume that the centroid of $K$ is the origin and retain the standing well-roundedness and oracle-encoding assumptions. The lower bound holds even when both bodies are centrally symmetric about the origin. For the general convex body containment problem in the oracle model, we establish in this section that the optimal approximation guarantee is tightly bounded by $s\in\Theta(d/\log d)$. Specifically, a result by \citet{Naszodi2019approximating} shows that, given a polynomial number of samples, one can approximate the containment problem within a factor $s\in O(d/\log d)$ with high probability. Our main contribution in this section is to prove a matching lower bound: any algorithm that accesses the oracle through a polynomial number of queries must incur an approximation error of at least $s\in\Omega(d/\log d)$ with at least constant probability. This establishes the tightness of the approximation in this oracle setting and further indicates that certain existential results on polytopal approximations of convex bodies cannot be made algorithmically feasible.

\subsection{An \texorpdfstring{$O(d/\log d)$}{O(d/log d)} Upper Bound}

To achieve the desired approximation guarantee for containment in general convex bodies, we first recall a result of \citet{Naszodi2019approximating} regarding polytopal approximation of convex bodies.

\begin{theorem}[Theorem 1.2 in \citep{Naszodi2019approximating}]\label{thm:polytope-general-convex}
Given a convex body $K\subseteq\R^d$ whose centroid is the origin and $1 < s \le d$, sampling points $x_1,\ldots,x_T$ independently and uniformly from $K$, where
\[
 T\in\Theta\!\left(d(1-1/s)^{-d}\log\!\bigl((1-1/s)^{-d}\bigr)\right),
\]
yields that the convex hull $P=\mathrm{conv}(x_1,\ldots,x_T)$ satisfies $P\subseteq K\subseteq sP$ with high probability. The implicit constant can be chosen to make the failure probability at most $1/\poly(d)$.
\end{theorem}

The containment guarantee now directly follows by approximating $K$ by $P$ in the sense of \Cref{thm:polytope-general-convex} with $s(d)\in\Theta(d/\log d)$.
The random samples can be found in polynomial time~\cite{DyerFK91}. More precisely, it suffices to generate samples whose joint distribution is within $\delta$ in total variation of the independent uniform law; independent calls with error at most $\delta/T$ each suffice, which increases the failure probability by at most $\delta$.

\begin{corollary}\label{cor:sampling-general-convex}
For convex bodies $K,Q\subseteq\R^d$ under the assumptions above and $s\in\Theta(d/\log d)$, $(K,Q)$\textup{\textsc{-Gap-Containment}} can be solved by a randomized polynomial time algorithm with high probability.
\end{corollary}

\subsection{A Matching Lower Bound}
\label{sec:matching-lower-bound}

We now prove that the $\Omega(d/\log d)$ scaling for the containment problem of general oracle-access convex bodies is tight.

\begin{theorem}\label{thm:hardness-general-cvx-bodies}
For $s\in o(d/\log d)$, there does not exist a randomized polynomial time algorithm that decides $(K,Q)$\textup{\textsc{-Gap-Containment}} for each $K,Q\subseteq\R^d$ with high probability. This holds even when $K$ and $Q$ are centrally symmetric about the origin.
\end{theorem}

To prove the above theorem, we first perform a reduction from hardness of computing the radius of a general convex body represented by an oracle. Below, we denote by $\textsc{OutRad}_2(K)$ the $\ell_2$-circumradius of a convex body $K\subseteq\R^d$. For an origin-symmetric body $P$, this equals $\max_{x\in P}\|x\|_2$. Note that below we use an optimization oracle, which can be simulated by polynomially many membership queries \cite{GrotschelLS84}.

\begin{lemma}\label{lem:hardness-radius-containment}
Fix $r>1$ and suppose that, for a family of origin-symmetric convex bodies $P\subseteq\R^d$, one cannot distinguish
\[
 \textup{\textsc{OutRad}}_2(P)\le1
 \qquad\text{from}\qquad
 \textup{\textsc{OutRad}}_2(P)>r
\]
with success probability at least $2/3$ using polynomially many calls to the input oracle. Then there are origin-symmetric convex bodies $K,Q\subseteq\R^d$ for which one cannot distinguish $K\subseteq Q$ from $K\nsubseteq r^2Q$ with success probability at least $2/3$ using polynomially many membership queries.
\end{lemma}

\begin{proof}
Given an origin-symmetric convex body $P$, set $K=P$ and $Q=P^\circ$. We prove that, for every $r\ge1$,
\[
 \textsc{OutRad}_2(P)\le r
 \quad\Longleftrightarrow\quad
 K\subseteq r^2Q.
\]
If $\textsc{OutRad}_2(P)\le r$, then $P\subseteq r\B_2^d$, so polarity gives $P^\circ\supseteq(1/r)\B_2^d$. Therefore,
\[
 K=P\subseteq r\B_2^d\subseteq r^2P^\circ=r^2Q.
\]
Conversely, suppose $P\subseteq r^2P^\circ$. For every $a\in P$, we have $a/r^2\in P^\circ$, and consequently
\[
 \frac{\|a\|_2^2}{r^2}
 =\scalar{\frac{a}{r^2},a}\le1.
\]
Thus $\|a\|_2\le r$ for every $a\in P$, giving $\textsc{OutRad}_2(P)\le r$. Applying the equivalence with $r=1$ and with the stated $r$ proves the reduction, with polynomially many oracle queries.
\end{proof}

We now show \Cref{thm:hardness-general-cvx-bodies}.

\begin{proof}[Proof of \Cref{thm:hardness-general-cvx-bodies}]
The radius lower bound of \cite{brieden1998approximation} gives a distribution over origin-symmetric convex bodies $P\subseteq\R^d$ for which one cannot distinguish
\[
 \textup{\textsc{OutRad}}_2(P)\le1
 \qquad\text{from}\qquad
 \textup{\textsc{OutRad}}_2(P)\in\Omega\!\left(\sqrt{\frac d{\log d}}\right)
\]
with success probability at least $2/3$ using polynomially many oracle calls, even with optimization access. The hard family consists of the Euclidean ball $P = \B_2^d$ and bodies $P = P_v=\operatorname{conv}(\B_2^d,\pm Rv)$, where $v$ is a hidden unit vector.
Suppose a containment algorithm achieves factor $s=o(d/\log d)$, and apply it to
\[
 (K,Q)=\left(\frac{1}{s}P,P^\circ\right).
\]
If $\textsc{OutRad}_2(P)\le1$, then $sK=P\subseteq P^\circ=Q$, so a \textsc{No} answer is invalid. In the large-radius case, $\textsc{OutRad}_2(P)>\sqrt{s}$ for sufficiently large $d$, so \Cref{lem:hardness-radius-containment} gives $K=1/s \, P\nsubseteq P^\circ=Q$, making \textsc{Yes} invalid. Thus the algorithm distinguishes the radius instances with the same success probability, a contradiction.
\end{proof}

%% file: new_sections/05-polytopes.tex
\section{Equivalence of Containment and Polytope Approximation}
\label{sec:containment-polytope}

In this section, we state reductions between containment and approximation by a polytope
specified by a list of points, deferring their proofs to \Cref{app:omitted-polytopes}. As always, $K\subseteq\R^d$ is a convex
body with $0\in\operatorname{int}K$, membership in $K$ is available, and
$s\ge1$ is a given rational number. We use the usual gap convention:
on input $(K,Q)$, a valid answer is \textsc{Yes} if $K\subseteq Q$, and
\textsc{No} if $sK\nsubseteq Q$, and either answer is allowed when both
conditions hold. As before, the outer body is accessed only through membership
queries and, when required, supplied inner and outer radius bounds.

\begin{restatable}{proposition}{propcontainmentpolytope}\label{prop:containment-polytope}
The following reductions preserve the approximation factor $s$.
\begin{enumerate}
\item Given $P=\conv\{v_1,\ldots,v_N\}$ with
      $P\subseteq K\subseteq sP$, one can solve
      $(K,Q)\textsc{-Gap-Containment}$ using $N$ membership queries to $Q$.
      If the approximation is produced by a randomized algorithm with
      failure probability at most $\delta$, the resulting containment
      algorithm has failure probability at most $\delta$.
\item Suppose a deterministic algorithm solves
      $(K,Q)\textsc{-Gap-Containment}$ for every outer convex body $Q$
      containing the origin in its interior, using at most $q$
      membership queries to $Q$. One execution of the algorithm suffices
      to construct a polytope $P$ with at most $q+1$ vertices satisfying
      $P\subseteq K\subseteq sP$.
\item Suppose instead that the algorithm is randomized, succeeds with
      probability at least $2/3$ for every such $Q$, and makes at most $q$
      membership queries in every execution. For every $\delta\in(0,1)$,
      one can construct a polytope $P$ with at most $T+1$ vertices, where
      \[
       T=\left\lceil Cq\left((d+1)\log(2q)+\log\left(\frac{1}{\delta}\right)\right)\right\rceil,
      \]
      such that $P\subseteq K$ always and $K\subseteq sP$ with probability
      at least $1-\delta$. The construction uses $T$ independent
      executions of the containment algorithm, and $C$ is a sufficiently
      large absolute constant.
\end{enumerate}
In the last two reductions, each outer-membership query is simulated
using one membership query to $K$. Consequently, all reductions run in
polynomial time when the underlying algorithms do, with polynomial
overhead in $d,q$, and $\log(1/\delta)$.
\end{restatable}

\paragraph{Implications for Approximation of Zonotopes by Polytopes.}

The reductions in \Cref{prop:containment-polytope} above together with \Cref{thm:anc-containment} show that the universal lower bound in \Cref{thm:appr_zono_by_poly_error} is tight for all zonotopes.

\begin{corollary}
There is a randomized polynomial-time algorithm that, given a
full-dimensional zonotope $Z\subseteq\mathbb R^d$, $d\ge3$, by its
generators, computes, with high probability, a polytope $P$ with
$\poly(d)$ many vertices satisfying
\[
 P\subseteq Z\subseteq O\!\left(\sqrt{\frac d{\log d}}\right)P.
\]
\end{corollary}
\begin{proof}
    Our containment algorithm (\Cref{alg:anc-run})
    uses polynomially many oracle queries in $d$, so the reduction yields,
    with high probability, a polytope $P$ with polynomially many vertices such
    that
    \[
     P\subseteq Z\subseteq O\!\left(\sqrt{\frac d{\log d}}\right)P.
    \]
    Consequently, containment approximation and volumetric approximation by
    polynomial-vertex inscribed polytopes have the same optimal order
    $\Theta(\sqrt{d/\log d})$ for zonotopes. \Cref{thm:anc-containment,thm:appr_zono_by_poly_error,prop:containment-polytope} conclude the proof.
\end{proof}

\paragraph{Implications for Approximation of Convex Bodies by Polytopes.}

We turn to the negative implications of \Cref{thm:hardness-general-cvx-bodies} for algorithmically approximating convex bodies with polytopes. Before discussing them, we first review key results on polytopal approximations of convex bodies. Beyond the algorithmic $O(d/\log d)$ approximation for centered convex bodies \citep{Naszodi2019approximating} (\Cref{thm:polytope-general-convex}), a recent information-theoretic lower bound establishes this factor as optimal \citep{huang2026hardness}. If we restrict our attention to centrally symmetric convex bodies, a renowned result of \citet{Barvinok14} shows that the convex hull of $\poly(d)$ points from $K$ approximates $K$ up to an $O(\sqrt d)$ factor.

\begin{theorem}[\citet{Barvinok14}]\label{thm:barvinok}
For any origin-symmetric convex body $K$, there exists a polytope $P$ with at most $\poly(d)$ many vertices such that $1/s\, P\subseteq K\subseteq P$, for $s\in O(\sqrt d)$.
\end{theorem}

We now highlight the key negative implication of \Cref{thm:hardness-general-cvx-bodies}. Its hard instances consist of two origin-symmetric convex bodies, so symmetry alone does not mitigate hardness for containment in the oracle model. Crucially, this shows that Barvinok's existential result (\Cref{thm:barvinok}) cannot be made efficiently algorithmic in a vertex representation: \Cref{cor:vertex-apx-symmetric} proves that such an algorithm would contradict the $\Omega(d/\log d)$ hardness for containment.

\begin{corollary}\label{cor:vertex-apx-symmetric}
For $s\in o(d/\log d)$, there does not exist a randomized polynomial time algorithm that, for every origin-symmetric convex body $K\subseteq\R^d$ given by a membership oracle, outputs an explicit list of polynomially many vertices of a polytope $P$ satisfying $1/s\,P\subseteq K\subseteq P$ with probability at least $2/3$. In particular, this rules out such an efficient algorithm for finding the polytope in \Cref{thm:barvinok}.
\end{corollary}

\begin{proof}
Assume for contradiction that such an algorithm exists. On input $(K,Q)$, compute its polytope $P$ and test all listed vertices for membership in $Q$. If all are accepted, return \textsc{Yes}; otherwise return \textsc{No}. On the successful approximation event,
\[
 \frac1sP\subseteq K\subseteq P.
\]
If $P\subseteq Q$, then $K\subseteq Q$. Otherwise, a rejected vertex belongs to $P\subseteq sK$, so $sK\nsubseteq Q$. Thus, on this event, the returned answer is valid for $(K,Q)$\textup{\textsc{-Gap-Containment}}. The gap algorithm succeeds with probability at least $2/3$, contradicting \Cref{thm:hardness-general-cvx-bodies}.
Therefore the existence result in \Cref{thm:barvinok} cannot be made efficiently constructive in vertex representation unless $s\in\Omega(d/\log d)$.
\end{proof}

We note that \Cref{cor:vertex-apx-symmetric} is related to the distinction
between existential and algorithmic ellipsoidal approximation.
For every origin-symmetric convex body $K\subseteq\R^d$, its
maximal-volume inscribed ellipsoid---the \emph{John ellipsoid}---satisfies
$E\subseteq K\subseteq\sqrt dE$.
Algorithmically, given an oracle for $K$ and known bounds
$r\B_2^d\subseteq K\subseteq R\B_2^d$, one can compute an origin-centered ellipsoid $F$ satisfying $F\subseteq K\subseteq O(d)F$ in oracle-polynomial time
\cite[Theorem~4.6.3]{DBLP:books/sp/GLS1988}.
Conversely, the spherical-cap argument in the proof of
\cite[Theorem~2.4]{brieden2001deterministic},
applied to hidden-spike bodies and their polars, rules out any randomized
algorithm that uses polynomially many separation-oracle queries and
guarantees $F\subseteq K\subseteq sF$ with probability at least $2/3$
for every origin-symmetric $K$, when $s\in o(d/\log d)$.
Thus the $O(d)$ oracle upper bound is optimal up to a logarithmic factor.

\begin{remark}[Existential lower bound for general symmetric convex bodies]
The $O(\sqrt{d/\log d})$ polynomial-vertex approximation for zonotopes does not extend to all origin-symmetric convex bodies, even existentially. This follows from a classical finite-field construction in \cite[Lemma~3.2]{DBLP:conf/soda/BadanidiyuruDFKNR12}. Note that this shows that Barvinok's result in \cite{Barvinok14} is tight for every fixed polynomial number of vertices. 

Fix a constant integer $k\ge1$, put $\ell=2k+2$. We allow the approximating polytope $P \subseteq K \subseteq sP$ to have $d^k$ vertices and we construct the body $K$ by adding sufficiently many spikes to $\B_1^d$ such that $s \geq \sqrt{d}/(2k+2)$ is necessary. For this, choose a prime $p>\ell$, and set $d=p^2$. Let $\mathcal S$ consist of the graphs of all polynomials of degree at most $\ell$ over $\mathbb F_p$. To be more precise, for every such polynomial $q$, $\mathcal S$ contains the set $S_q = \{(0,q(0)),\dots,(p-1,q(p-1))\}$. Then $|\mathcal S|=p^{\ell+1}$, every set has size $p$, and distinct sets intersect in at most $\ell$ points, since two distinct polynomials of degree at most $\ell$ can coincide in at most $\ell$ positions. Identifying the coordinates with $\mathbb F_p^2$, define the symmetric convex body
\begin{align*}
     K=\conv\bigl(\{\pm e_i:i\le d\}
       \cup\{\pm\mathbf1_S:S\in\mathcal S\}\bigr) = \conv\bigl(\B_1^d \cup \{\pm\mathbf1_S:S\in\mathcal S\}\bigr).
\end{align*}
Suppose $P\subseteq K\subseteq sP$ and $P$ has at most $d^k$ vertices. By Carathéodory's theorem, each vertex is a convex combination of at most $d+1$ defining points of $K$. Since $|\mathcal S|>d^k(d+1)$, some $S\in\mathcal S$ has neither $\mathbf1_S$ nor $-\mathbf1_S$ used in these representations. Every remaining defining point has inner product at most $\ell$ with $\mathbf1_S$. Consequently,
\begin{align*}
     p= \langle \mathbf{1}_S, \mathbf{1}_S \rangle = h_K(\mathbf1_S)\le s\,h_P(\mathbf1_S)\le s\ell,
\end{align*}
and hence, $s\ge \sqrt d/({2k+2})$. Thus, for every fixed polynomial vertex budget, there are arbitrarily large dimensions and symmetric convex bodies requiring an $\Omega(\sqrt d)$ approximation factor.
Finally, we remark that the body $\B_1^d$ is only part of the definition of $K$ to ensure its full-dimensionality. Indeed, $\B_1^d$ can be replaced by any sufficiently small full-dimensional convex body, and the argument still holds.
\end{remark}

%% file: new_sections/appendix.tex
\appendix

\section{Omitted Proof from \Cref{sec:cont_normal_cone}}\label{app:zonotope-tails}

\zonotopetails*

\paragraph{Proof overview.} The proof follows the argument of \Cref{prop:clipped-run}, with a
few changes to make the probability bound uniform over all heights.
We retain the notation of \Cref{sec:clipped-run}. For
$1\le\tau\le2\sqrt d$, run \Cref{alg:anc-run} with
$q(t)=1/(1+|t|)$ and initial scale $\lambda_0=\log(2+\tau)$,
leaving all other instructions unchanged. In the analysis, we
replace $\sqrt{\log d}$ by $\tau$ in the conditioning event.
The number of iterations now depends on $\tau$ rather than $d$,
which is important for obtaining constant probabilities when
$\tau$ is bounded. At the end, we randomly choose the scale,
giving a single sampling distribution that works at every height.

Write $x_{{\rm orig},j}=A_jx_j$ for the ordinary outputs and let
$R$ be the number of iterations. After the run, form the partial
sums $\widehat x_k=\sum_{j<k}x_{{\rm orig},j}$ for
$k=1,\ldots,R$, with $\widehat x_0=0$. Their role in the proof
is to make the sample heights in the support-loss bounds telescope
into the height of one partial sum. This avoids accumulating a
separate error in each iteration, which the original argument
does not control uniformly over the full range of heights.

Every partial sum belongs to $Z$. Indeed, the coefficient of
$w_i$ in $A_jv_i^{(j)}$ is $\prod_{\ell<j}q_{i,\ell}$.
Since $|tq(t)|=1-q(t)$, its absolute coefficient in
$\widehat x_k$ is at most
\[
 \sum_{j<k}\left(\prod_{\ell<j}q_{i,\ell}\right)(1-q_{i,j})
 =1-\prod_{j<k}q_{i,j}\le1.
\]
As in \Cref{sec:clipped-run}, we also have that these weights satisfy $H_j\le\sqrt d\,\|b_j\|$ and $\|b_{j+1}\|\le\|b_j\|$.

Finally, moment estimates replace the reverse-Markov step used
in the original proof, avoiding its dimension-dependent loss
when converting a large conditional mean into a success
probability. The Gaussian conditioning and norm-decay argument
remain the same.

\begin{lemma}[Moment estimates]\label{lem:tail-moments}
For the local sample $x$ at scale $L=C\lambda\ge C$, if $\mathcal L(V)\preceq I$, then $\mathbb E\langle b,x\rangle^2\le\|b\|^2/L^2$. If also $\mathcal L(V)\succeq I/4$, then, for an absolute $c_1>0$,
\[
 \P{}{\langle b,x\rangle\ge\frac{\|b\|}{8L}}\ge c_1.
\]
\end{lemma}

\begin{proof}
The map $\langle b,x\rangle$ is odd and $\|b\|/L$-Lipschitz, since $(tq(t))'=q(t)^2$ and $\mathcal L(V)\preceq I$. Gaussian Poincar\'e gives the second-moment bound. For the other claim, assume $b\ne0$ and put $u=\langle z,b/\|b\|\rangle$. Gaussian conditioning gives
\[
 \mathbb E[\langle b,x\rangle u]
 =\frac{b^\top\mathcal L(V)b}{L\|b\|}\,
       \mathbb E[N^2q(N/L)]
 \ge\frac{3\|b\|}{16L},
\]
where $N\sim N(0,1)$; we used $q(t)\ge1-|t|$ and $L\ge C\ge64$. Cauchy--Schwarz therefore gives $\mathbb E\langle b,x\rangle^2\ge9\|b\|^2/(256L^2)$. Applying Poincar\'e inequality to $\langle b,x\rangle^2$ also gives $\mathbb E\langle b,x\rangle^4\le5\|b\|^4/L^4$. Paley--Zygmund and symmetry then prove the claim.
\end{proof}

\begin{lemma}[Good updates]\label{lem:tail-good}
Define $\mathcal E$ as in \Cref{lem:good_update}, replacing $\sqrt{\log d}$ by $\tau$, and retain $\mu=\mathbb E[\langle b,x\rangle\mid\mathcal E]$. Then $\mu>0$ and $\P{}{\mathcal E}\ge c_0\exp{-\tau^2/(C^2\lambda^2)}$. Conditional on $\mathcal E$, with probability at least $1/2$, both inequalities
\begin{equation}\label{eq:tail-good}
 \|b^+\|^2\le\frac{32C^2\lambda^2}{\tau}\|b\|\mu,
 \qquad H-H^+\le\frac{4H}{C\lambda}+\langle b,x\rangle
\end{equation}
hold. We call such an update good.
\end{lemma}

\begin{proof}
We use the same conditional-Gaussian argument as in \Cref{lem:good_update}. The scalar estimates replacing \Cref{lem:clipped-scalar} are, for $Y=m+\sigma N$, $m\ge0$, $0\le\sigma\le1$, and $L\ge1$,
\[
 m\,\mathbb E q(Y/L)\le8\,\mathbb E[Yq(Y/L)],
 \qquad
 \mathbb E[(1-q(Y/L))\mathbf1_{\{Y<0\}}]\le\frac1{2L}.
\]
Indeed, integrating the derivative of $yq(y/L)$ and using Jensen's inequality yields
\[
 \mathbb E[Yq(Y/L)]
 \ge\int_0^m\frac{L^2}{(L+u+1)^2}\,du
 =\frac{mL^2}{(L+1)(L+m+1)}
 \ge\frac{mL}{4(L+m)}.
\]
The triangle inequality gives $(L+m)q(Y/L)\le L+|\sigma N|$, proving the first estimate. The second follows from $1-q(t)=|t|q(t)$ and $\mathbb E(-Y)_+\le1/\sqrt{2\pi}<1/2$.

As before, orient the generators so that $\langle b,v_i\rangle\ge0$ and condition on $u=\langle z,b/\|b\|\rangle$. Applying the first scalar estimate and averaging over $\mathcal E$ gives $\mathbb E[\|b^+\|^2\mid\mathcal E]\le8C^2\lambda^2\|b\|\mu/\tau$. For the support, use
\[
 H-H^+=\langle b,x\rangle
       +2\sum_{i=1}^n(1-q_i)\langle b,v_i\rangle\mathbf1_{\{t_i<0\}}.
\]
The additional term is nonnegative and, by the second scalar estimate, has conditional expectation at most $H/(C\lambda)$. Markov's inequality with factor four and a union bound give \eqref{eq:tail-good}. Positivity of $\mu$ and the probability bound for $\mathcal E$ follow exactly as in \Cref{lem:good_update}.
\end{proof}

\begin{lemma}[Single-run guarantee]\label{lem:tail-run}
There exist absolute constants $c,B>0$ such that, for every fixed $a\ne0$, with probability at least $c \, \exp{-B\tau^2}$, one of the partial sums satisfies $\langle a,\widehat x_k\rangle\ge c\tau h_Z(a)/\sqrt d$.
\end{lemma}

\begin{proof}
Normalize $h_Z(a)=1$ and fix a sufficiently small absolute $\eta>0$. If $\|b_0\|\ge4C\eta\tau\lambda_0/\sqrt d$, the first output already reaches height $\eta\tau/(2\sqrt d)$ with constant probability by \Cref{lem:tail-moments}. Otherwise $\|b_j\|\le4C\eta\tau\lambda_0/\sqrt d$ throughout the run. Call a state small-mean if $\mu_j<\eta\tau/\sqrt d$.

Suppose all states are small-mean, all updates are good, and no partial sum has absolute height at least $\eta\tau/(4\sqrt d)$. As in the proposition, $\lambda_j=\lambda_0/2^j$ and $\sum_{j<R}\lambda_j^{-1}<2$. The support bounds now telescope:
\[
 H_R\ge1-\frac8C-\sum_{j<R}\langle b_j,x_j\rangle
 =1-\frac8C-\langle a,\widehat x_R\rangle
 \ge\frac34.
\]
On the other hand, iterating the norm bounds as in the proof of \Cref{prop:clipped-run} gives
\[
 H_R\le\sqrt d\,\|b_R\|
 \le\lambda_R^2(512C^2\eta)^{1-2^{-R}}
    \left(\frac{\sqrt d\,\|b_0\|}{\lambda_0^2}\right)^{2^{-R}}
 \le64C\sqrt\eta<\frac12.
\]
Here $\sqrt d\,\|b_0\|/\lambda_0^2\le4C\eta\tau/\lambda_0\le\tau$, $\lambda_R<1$, and $2^{-R}\log\tau<1$. We choose $\eta$ so that $4C\eta\le1$ and $512C^2\eta\le1$, as well as the displayed strict inequality. Thus such a run must instead produce a partial sum of the required absolute height.

At a large-mean state, splitting the expectation at $\mu/2$ and applying Cauchy--Schwarz gives
\[
 \P{}{\mathcal E,\ \langle b,x\rangle\ge\mu/2}
 \ge\frac{\mu^2\P{}{\mathcal E}^2}{4\mathbb E\langle b,x\rangle^2}
 \ge\frac{\P{}{\mathcal E}^2}{64\lambda_0^2}.
\]
We used $\mu\ge\eta\tau/\sqrt d$, \Cref{lem:tail-moments}, the bound on $\|b_j\|$, and $\lambda_j\ge1$. Such an increment forces one of its adjacent partial sums to have absolute height at least $\mu/4$. Therefore, by the same adaptive probability argument as in \Cref{prop:clipped-run}, the probability of reaching the required absolute height is at least
\[
 \prod_{j<R}\left[\frac{\kappa}{\lambda_0^2}
          \exp{-\frac{2\tau^2}{C^2\lambda_j^2}}\right]
 \ge\frac{\kappa^R}{\lambda_0^{2R}}
          \exp{-\frac{8\tau^2}{3C^2}}
 \ge c \, \exp{-B\tau^2},
\]
where $\kappa>0$ is absolute. We used $\sum_{j<R}\lambda_j^{-2}\le4/3$ and $R=O(\log\log \tau)$. Reversing all Gaussian draws leaves the weights and coordinate
updates unchanged and replaces every $\widehat x_k$ by
$-\widehat x_k$. Thus the positive- and negative-height events
have equal probability. Their union is the absolute-height event,
so a union bound gives
\[
 \P{}{\max_{1\le k\le R}\langle a,\widehat x_k\rangle
          \ge\frac{\eta\tau}{4\sqrt d}}
 \ge\frac12
 \P{}{\max_{1\le k\le R}|\langle a,\widehat x_k\rangle|
          \ge\frac{\eta\tau}{4\sqrt d}}
 \ge\frac{c_2}{2}\exp{-B\tau^2}.
\]
In the initial large-norm case, \Cref{lem:tail-moments} already gives
\[
 \P{}{\langle a,\widehat x_1\rangle
          \ge\frac{\eta\tau}{2\sqrt d}}\ge c_1.
\]
Since $\exp{-B\tau^2}\le1$, both cases establish the lemma with
$c=\min\{\eta/4,c_1,c_2/2\}$.
\end{proof}

\begin{proof}[Proof of \Cref{cor:zonotope-tails}]
To choose the scale, start at $\tau=1$ and toss a fair coin: stop
on heads, and otherwise double $\tau$ and repeat. Stop in any case
once $\tau\ge\sqrt d$. Run the construction at the chosen scale
and return one of the partial sums $\widehat x_1,\ldots,\widehat x_R$
uniformly at random. Both choices are independent of the Gaussian
draws. This procedure takes polynomial time and always returns a
point in $Z$.

Now fix $1\le s\le\sqrt d$. Among the possible scales, the first
one at least $s$ satisfies $s\le\tau<2s$ and is chosen with
probability at least $1/(4s)$. Conditional on choosing this scale,
\Cref{lem:tail-run} shows that the returned point reaches height
at least $c\tau h_Z(a)/\sqrt d$ with probability at least
$c/R \, \exp{-B\tau^2}$. Since $R=O(\log\log\tau)$, the factors
$1/R$ and $1/(4s)$ can be absorbed into $\exp{-O(s^2)}$, giving
the claimed bound. For $0\le s\le1$, choosing $\tau=1$, which
happens with probability $1/2$, already gives a constant lower
bound.
\end{proof}

\section{Omitted Proofs from \Cref{sec:containment-polytope}}\label{app:omitted-polytopes}

We now restate the reductions between polytope approximation and containment for convenience, and then prove them. To do so, we use the following standard consequence of the $\varepsilon$-net theorem
for halfspaces, whose VC dimension is $d+1$
\citep[Section 3]{Naszodi2019approximating}. If $\mu$ is a probability
measure on $\R^d$, $\varepsilon\in(0,1/2]$, $\delta\in(0,1)$ and $C$ is an absolute constant, then
\[
 N\ge\frac{C}{\varepsilon}
       \left((d+1)\log\left(\frac{2}{\varepsilon}\right)+\log\left(\frac{1}{\delta}\right)\right)
\]
independent samples from $\mu$ intersect every closed halfspace of
$\mu$-measure at least $\varepsilon$, simultaneously with probability at
least $1-\delta$.

\propcontainmentpolytope*

\begin{proof}
We proceed in order of the statements.\medskip

\noindent \emph{1. From polytope approximation to containment.}
Query whether $sv_i\in Q$ for every $i$. If one query is rejected, then
$sv_i\in sP\subseteq sK$ certifies $sK\nsubseteq Q$, so return
\textsc{No}. Otherwise, convexity gives $sP\subseteq Q$, and hence
$K\subseteq Q$, so return \textsc{Yes}. The argument applies on every
successful realization of a randomized polytope approximation, proving
the failure-probability assertion.\medskip

\noindent \emph{2. The deterministic converse.}
Run the containment algorithm against $Q_0=sK$, simulating a query $x$
by testing $x/s\in K$. Since $sK\subseteq Q_0$, the only valid answer
is \textsc{Yes}. Let $B$ consist of the query points accepted by $Q_0$,
and put $\widehat P=\conv(\{0\}\cup B), P=1/s \, \widehat P$. Clearly $P\subseteq K$, and the displayed list has at most $q+1$ points.
Suppose for the sake of contradiction that $K\nsubseteq\widehat P$, then separation gives $a\ne0$ with
$h_{\widehat P}(a)<h_K(a)$. Now, choose a $t$ such that $\max\{h_{\widehat P}(a),h_K(a)/2\}<t<h_K(a)$ and define
\[
 Q_{a,t}=sK\cap\{x:\langle a,x\rangle\le t\}.
\]
This is a convex body containing the origin in its interior. Every query $y$
accepted in the execution on $Q_0$ belongs to $B \subseteq \widehat P$ and so $\scalar{a, y} \le h_{\widehat P}(a) < t$, which means that $y \in Q_{a,t}$. If $y$ is rejected, then $y\notin Q_0$ and therefore
$y\notin Q_{a,t}$, since $Q_{a,t}\subseteq Q_0$. This means that for a fixed query the executions on $Q_0$ and $Q_{a,t}$ always agree, whether on a \textsc{Yes} or on a \textsc{No}. Induction over the queries shows that the
deterministic executions have identical transcripts and outputs.
Recall, however, that on $(K,Q_0)$, the answer \textsc{No} is invalid because
$sK=Q_0$. On the other hand, compactness of $K$ gives a point
$x^\star\in K$ satisfying $\langle a,x^\star\rangle
 =h_K(a)>t\ge h_{Q_{a,t}}(a)$. In particular, $x^\star\notin Q_{a,t}$, so
$K\nsubseteq Q_{a,t}$ and the oracle answer is \textsc{No}.
This is a contradiction because the only valid answer for $Q_0$ (and thus for $Q_{a,t}$) is \textsc{Yes}. This, therefore, proves $K\subseteq\widehat P=sP$.\medskip

\noindent \emph{3. The randomized converse.}
As above, run the algorithm against $Q_0=sK$. For its $j$th query $y_j$,
record $x_j=y_j/s \in K$ if it is accepted and $x_j=0 \in K$ otherwise,
padding with zeros up to $q$ points. Choose $J$ uniformly from $\{1,\ldots,q\}$, independently
of the execution, and let $\mu$ be the law of $X=x_J$.

We next fix $a\ne0$ and $h_K(a)/2\le t<h_K(a)$, and couple the runs on
$Q_0$ and $Q_{a,t}$ using the same random tape.
Each run gives an incorrect output with probability at most $1/3$.
By a union bound, the probability that either is incorrect is at
most $2/3$. Thus, with probability at least $1/3$, both are correct
and return \textsc{Yes} and \textsc{No}, respectively.
By the transcript argument in the deterministic converse,
different outputs require a query $y_j$ accepted by $Q_0$ with
$\langle a,y_j\rangle>t$. Its recorded point satisfies
$\langle a,x_j\rangle>t/s$, so
$\Pr(\max_j\langle a,x_j\rangle>t/s)\ge1/3$.

As $t$ increases to $h_K(a)$, these events decrease to
$\{\max_j\langle a,x_j\rangle\ge h_K(a)/s\}$, whose probability
is therefore still at least $1/3$. Hence, for the closed halfspace
$F_a=\{x:\langle a,x\rangle\ge h_K(a)/s\}$,
\[
 \mu(F_a)
 =\frac1q\sum_{j=1}^q\Pr(x_j\in F_a)
 \ge\frac1q\Pr(\exists j:x_j\in F_a)
 \ge\frac1{3q}.
\]
Now generate independent samples $X_1,\ldots,X_T$ from $\mu$,
using a fresh execution and index $J$ for each, and put
$P=\conv\{0,X_1,\ldots,X_T\}$. Since all samples lie in $K$,
we always have $P\subseteq K$.
The halfspace $\varepsilon$-net theorem guarantees that a sample
of the stated size meets every halfspace of $\mu$-measure at least
$\varepsilon$, simultaneously, with probability at least $1-\delta$.
Taking $\varepsilon=1/(3q)$ and a sufficiently large constant
in the definition of $T$, the sample therefore meets every $F_a$.
This gives $h_P(a)\ge h_K(a)/s$ for every $a\ne0$, and hence
$K\subseteq sP$.
\end{proof}